\documentclass{llncs}
\usepackage{citesort}
\usepackage{graphicx}
\usepackage{theorem}
\usepackage{latexsym}
\usepackage{multirow}
\usepackage{algorithm,algorithmic}
\usepackage{amssymb} 
\usepackage{amsmath} 
\usepackage{multirow,eepic}
\usepackage{wrapfig}
\usepackage{lipsum}

\makeatletter
\def\senbun#1(#2)#3({\@senbun(#2)(}
\def\@senbun(#1,#2)(#3,#4){%
   \@tempdima#1\p@ \advance\@tempdima#3\p@
   \divide\@tempdima\tw@
   \@tempdimb#2\p@ \advance\@tempdimb#4\p@
   \divide\@tempdimb\tw@
   \edef\@senbun@temp{\noexpand\qbezier(#1,#2)%
      (\strip@pt\@tempdima,\strip@pt\@tempdimb)(#3,#4)}%
   \@senbun@temp}
\makeatother
\newcommand{\EL}{\overset{\lambda}{\rightarrow}}
\newcommand{\EO}{\overset{1}{\rightarrow}}
\newcommand{\EZ}{\overset{0}{\rightarrow}}
\newcommand{\EH}{\overset{1/2}{\rightarrow}}
\newcommand{\ML}{${\mathcal L}$}
\newcommand{\N}{{\rm I\kern-.22em N}} 
\newcommand{\Z}{{\sf Z\kern-.42em Z}} 
\newcommand{\R}{{\rm I\kern-.22em R}} 
\newcommand{\BbbC}{{\rm\kern.22em\rule[.1ex]{.06em}{1.4ex}\kern-.28em C}} 
\newcommand{\BbbQ}{{\rm\kern.22em\rule[.1ex]{.06em}{1.4ex}\kern-.28em Q}}

\newcommand{\QED}{\mbox{}\hfill$\Box$}
\newcounter{Codeline}
\newcommand{\Newcodeline}{\setcounter{Codeline}{1}}
\newcommand{\Cl}{{\theCodeline}: \addtocounter{Codeline}{1}}
\newcommand{\crm}{\\}

\authorrunning{T. Okumura et.al.}
\titlerunning{Optimal Rendezvous for Asynchronous Mobile Robots with Lights}

\begin{document}

\title{Optimal Rendezvous ${\mathcal L}$-Algorithms \\
for Asynchronous Mobile Robots \\with External-Lights
%\thanks{The authors thank Prof. Defago}
}
\author{Takashi~OKUMURA\inst{1} \and Koichi~WADA\inst{2} \and Xavier D\'{E}FAGO\inst{3}}
\institute{Graduate School of Science and Engineering, Hosei University, \\ 
Tokyo 184-8584 Japan.\\
\email{takashi.okumura.4e@stu.hosei.ac.jp}
\and
Faculty of Science and Engineering, Hosei University, \\Tokyo, 184-8485, Japan.\\ 
\email{wada@hosei.ac.jp}
\and
School of Computing, Tokyo Institute of Technology, \\ Tokyo, 152-8550, Japan.\\
\email{defago@c.titech.ac.jp}
}
 
\maketitle

\begin{abstract}
%%%must changed
We study the {\em Rendezvous} problem for $2$ autonomous mobile robots in asynchronous settings
with persistent memory called {\em  light}.
It is well known that Rendezvous is impossible in a basic model when robots have no lights,
even if the system is semi-synchronous. On the other hand,
Rendezvous is possible if robots have lights of various types with a constant number of colors \cite{FSVY,V}.
If robots can observe not only their own lights but also other robots' lights, their lights are called {\em full-light}.
If robots can only observe the state of other robots' lights, the lights are called {\em external-light}.

%In asynchronous settings, Rendezvous can be solved by robots with $2$ colors of full-lights \cite{HDT},
In this paper, we focus on robots with external-lights in asynchronous settings and a particular class of algorithms (called \ML-algorithms),
where an \ML-algorithm computes a destination based only on the current colors of observable lights.
When considering \ML-algorithms, Rendezvous can be solved by robots with full-lights and $3$ colors in general asynchronous settings (called ASYNC) and
the number of colors is optimal under these assumptions. In contrast, 
there exists no \ML-algorithms in ASYNC with external-lights regardless of the number of colors \cite{FSVY}. 
In this paper, we consider a fairly large subclass of ASYNC in which Rendezvous can be solved by \ML-algorithms using external-lights with a finite number of colors, and we show that the algorithms are optimal in the number of colors they use.

\end{abstract}

\section{Introduction}
%\vspace{-2mm}
%{\small
. % no need for a funny story anymore. The dot ('.') is to circumvent a bug in latexdiff

\noindent{\bf Background and Motivation}\ \ 

The computational issues of autonomous mobile robots have been the object of much research in the field of distributed computing.
In particular, a large amount of work has been dedicated to the research of theoretical models of autonomous mobile robots
\cite{AP,BDT,CFPS,DKLMPW,IBTW,KLOT,SDY,SY}.
%FPS-book, etc
In the basic common setting, a robot is modeled as a point in a two dimensional plane and its capability is quite weak.
We usually assume that robots are {\em oblivious} (no memory to record past history), {\em anonymous} and {\em uniform} (robots have no IDs and run identical algorithms) \cite{FPS}.
Robots operate in $\mathit{Look}$-$\mathit{Compute}$-$\mathit{Move}$ ($\mathit{LCM}$) cycles in the model. In the Look operation, robots obtain a snapshot of the environment and
they execute the same algorithm using the snapshot as input for the Compute operation, and move towards the computed destination in the Move operation.
Repeating these cycles, all robots collectively perform a given task.
The weak capabilities of the robots make it challenging for them to accomplish even simple tasks.
Therefore, identifying the minimum (weakest) capabilities that the robots need to complete a given task in a given model constitutes a very interesting and important challenge for the theoretical research on autonomous mobile robots.

%In this paper, we also explore such weakest capabilities. In particular,
%we reveal the weakest additional assumptions for the task which cannot be solved in the basic common models. 
This paper considers the problem of {\em Gathering}, which is one of the most fundamental tasks for autonomous mobile robots.
Gathering is the process where $n$ mobile robots, initially located at arbitrary positions, meet within finite time at a location, not known a priori.
When there are two robots in this setting (i.e., for $n=2$)), this task is called {\em Rendezvous}. 
In this paper, we focus on Rendezvous in asynchronous settings and we reveal the relationship among several assumptions.

Since Gathering and Rendezvous are simple but essential problems, they have been intensively studied  
and a number of possibility and/or impossibility results have been shown under the different assumptions  \cite{AP,BDT,CFPS,DGCMR,DKLMPW,DP,FPSW,IKIW,ISKIDWY,KLOT,LMA,P,SDY}.
%A lot of refs.
The solvability of Gathering and Rendezvous depends on the activation schedule and the synchronization level. 
Usually three basic types of schedulers are identified, namely, the fully synchronous (FSYNC), the semi-synchronous (SSYNC) and the asynchronous (ASYNC) models.
In the FSYNC model, there is a common round and in each round all robots are activated simultaneously 
and $\mathit{Compute}$ and  $\mathit{Move}$ are done instantaneously. 
The SSYNC model is the same as FSYNC except that at each round only a subset of the robots are activated, with a fairness guarantee that every robot is activated infinitely-often in any infinite execution.
In the ASYNC scheduler, there are no restrictions about the notion of time. In particular, 
$\mathit{Compute}$ and $\mathit{Move}$ and the interval between them can take any (finite) duration, a robot can be seen while moving, and in the interval between an observation and a corresponding move other robots may have possibly moved several times. 
Gathering and Rendezvous are trivially solvable in FSYNC in the basic model (e.g., without lights) by using an algorithm that moves to the center of gravity. 
However, these problems can not be solved in SSYNC without any additional assumptions \cite{FPS}.

%Introduction of lights
Das \emph{et al.} \cite{DFPSY} extend the classical model with persistent memory, called {\em lights}, to reveal the relationship between ASYNC and SSYNC and they show that asynchronous robots equipped with lights and a constant number of colors, are strictly more powerful than semi-synchronous robots without lights.
%Rendzvous  
In order to solve Rendezvous without any other additional assumptions, robots with lights have been introduced \cite{FSVY,DFPSY,V}.
Table~\ref{tab:Table-Rendezvous} shows previous results including ours to solve Rendezvous by robots with lights, for each scheduler and movement restriction. 
In the table, $\mathit{LC}$-atomic ASYNC is a subclass of ASYNC, in which 
we consider from the beginning of each $\mathit{Look}$ operation to the end of the corresponding $\mathit{Compute}$ operation as an atomic one, that is,  no robot can observe between the beginning of each $\mathit{Look}$ operation and the end of the next $\mathit{Compute}$ on the same robot \cite{OWK}. 
Regarding the various kind of lights,
{\em full-light} means that robots can see their own light as well as that of the other robots, 
whereas {\em external-light} and {\em internal-light} %\footnote{In \cite{FSVY}, external-light and internal-light are called FCOMM and FSTATE, respectively.}  
respectively mean that they can see only the lights of the other robots, or only their own light. 
Regarding the movement restriction, 
{\em Rigid} means that the robots can always reach the computed destination during the move operation.
{\em Non-Rigid} means that robots may be stopped before reaching the computed destination but move a minimum distance $\delta>0$.
Non-Rigid(+$\delta$) means it is Non-Rigid and robots know the value $\delta$.
The Gathering of robots with lights is discussed in \cite{TWK}.

%L-algorithm 要修正
{\tiny
\begin{table}[h]
%\centering
\caption{Rendezvous algorithms by robots with lights.}
\label{tab:Table-Rendezvous}
{\footnotesize
\begin{tabular}{|c|c|c|c|c|c|}
\hline
scheduler      & movement  & full-light & external-light & internal-light & no-light \\ \hline\hline
% FSYNC
FSYNC & Non-Rigid%\senbun(0,10)(63,-4) 
                            & $-$ & $-$ & $-$ & $\bigcirc$ \cite{FPS}  \\ \hline\hline
% SSYNC
& Non-Rigid     & {\bf 2}$*$(S) \cite{V}& {\bf 3}$*$(S) \cite{V}       & $\infty *$ \cite{FSVY}       & \multirow{3}{*}{$\times$ \cite{FPS}} \\ \cline{2-5} 
     SSYNC          & Rigid & $-$    &  $-$         & 6 \cite{FSVY} &      \\ \cline{2-5} 
               & Non-Rigid(+$\delta$) &  $-$  &  $-$ & 3 \cite{FSVY} & \\ \hline\hline
% LC-atomic ASYNC
%\multirow{2}{*}{LC-atomic ASYNC}          & Non-Rigid     & {\bf 2}*(S) \cite{OWK}   & ?$\rightarrow$  {\bf 4}*(QS),{\bf 5}*(S)    & ?        & \multirow{2}{*}{$-$  }  \\ \cline{2-5} 
         & Non-Rigid     & {\bf 2}*(S) \cite{OWK}   & ?$\rightarrow$  {\bf 4}*(QS),{\bf 5}*(S)    & ?        & \multirow{2}{*}{$-$  }  \\ \cline{2-5} 
    $LC$-atomic            & Rigid & $-$     & ? $\rightarrow$ {\bf 3}* & ?      &     \\ \cline{2-5}
   ASYNC             & Non-Rigid(+$\delta$) & $-$   & ? & ? & \\ \hline\hline
% ASYNC
          & Non-Rigid     & {\bf 2}(S) \cite{HDT},3*(S) \cite{V}    & $\infty *$ \cite{FSVY}      & ?        & \multirow{2}{*}{$-$}  \\ \cline{2-5} 
         ASYNC     & Rigid & {\bf 2}* \cite{V}      & 12 \cite{FSVY}  & ?      &     \\ \cline{2-5}
                & Non-Rigid(+$\delta$) & $-$    & 3 \cite{FSVY} & ? & \\ \hline
\end{tabular}\\
\\
$\bigcirc$: solvable, $\times$: unsolvable. 
$*$: \ML-algorithm, (S): self-stabilizing, \\(QS): quasi-self-stabilizing.
$-$ indicates that  this part has been solved under weaker conditions or unsolved under stronger ones.
A number represents the number of colors used in these algorithms and it is in \textbf{boldface} when optimal.
$?$ means that this part has not been solved.
}%footnotesize
\end{table}
}

In Table~\ref{tab:Table-Rendezvous}, we can see that complete solutions have been obtained for the case of full-lights.
However, the cases of external-lights and internal-lights are still insufficiently explored and should be solved.

\noindent{\bf Our Contribution}\ \ 
%%% must changed
In this paper, we are concerned with Rendezvous for robots equipped with external-lights and a particular class of algorithms called \ML-algorithms. Briefly, an \ML-algorithm means that each robot
(1)~always computes a destination on the line connecting the two robots, and
(2)~using only the observed colors of the lights of the robots.

Algorithms of this class are of interest because they operate also
when the coordinate system of a robot is not self-consistent (i.e., it can unpredictably rotate, change its scale or undergo a reflection) \cite{FSVY}.
Rendezvous can be solved by an \ML-algorithm with $3$ colors of external-lights in SSYNC \cite{V}, but
cannot be solved by any \ML-algorithm with any number of colors of external-lights in ASYNC \cite{FSVY}.
%self-stabilizing

%$LC$-atomic 
In this paper, we reveal the relationship among the number of colors, movement restrictions and initial configurations
on \ML-algorithms with external-lights in $\mathit{LC}$-atomic ASYNC. In fact, 
we give three \ML-algorithms with external-lights in $\mathit{LC}$-atomic ASYNC, such that
(1) if we may start from a particular initial configuration with the same color, Rendezvous is solved with $3$ colors in Rigid, 
(2) if we start from any initial configuration with the same color (called {\em quasi-self-stabilizing}), Rendezvous is solved with $4$ colors in Non-Rigid, and (3) if we start from any initial configuration (called {\em self-stabilizing}), Rendezvous is solved with $5$ colors and in Non-Rigid.   
We also show that the numbers of colors used in the three algorithms are optimal in the sense that
no \ML-algorithm with fewer colors can solve Rendezvous. In order to derive the lower bounds 
we give several essential properties of \ML-algorithms.

The remainder of the paper is organized as follows. In Section
\ref{sec:model}, we define the robot model with lights, 
the Rendezvous problem, and basic terminology. 
Section \ref{sec:PRR} reviews 
previous results on Rendezvous with external-lights.
Section \ref{sec:SSYNCRendezvousAlgorithms} shows
several properties of \ML-algorithms for Rendezvous with $3$ colors of external-lights and
Section \ref{ASYNCRendezvousAlgorithms} shows optimal Rendezvous 
\ML-algorithms on Asynchronous robots with lights. 
Section \ref{sec:conclusion}
concludes the paper.

\section{Preliminaries}\label{sec:model}

\subsection{Robot Model}

We consider a set of $n$ anonymous mobile robots ${\cal R} = \{ r_1, \ldots, r_n \}$ located in $\R^2$.
Each robot $r_i$ has a persistent state $\ell(r_i)$ called {\em  light} which may be taken from a finite set of colors $L$. 

%$\N$
We denote by $\ell(r_i,t)$ the color that the light of robot $r_i$ has at time $t$ and 
$p(r_i, t) \in \R^2$ the position occupied by $r_i$ at time $t$ represented in some global coordinate system. 
Given two points $p,q \in \R^2$, $dis(p,q)$ denotes the distance between $p$ and $q$.

%A {\em configuration} ${\cal C}(t)$ at time $t$ is a multi-set of $n$ pairs $(\ell_i(t),p_i(t))$, each defining the color of light and the position of robot $r_i$ at time $t$. 
%When no confusion arises, ${\cal C}(t)$ is simply denoted by $\cal C$.

%For a subset $S$ of $L \times \R^2$, ${\cal L}(S)$ and ${\cal P}(S)$ are denoted as projections to $L$ and $\R^2$ from $S$, respectively.

Each robot $r_i$ has its own coordinate system where $r_i$ is located at its origin at any time. 
These coordinate systems do not necessarily agree with those of other robots. 
It means that there is no common knowledge of unit of distance, directions of its coordinates, or clockwise orientation ({\em chirality}).

At any point of time, a robot can be active or inactive. When a robot $r_i$ is activated, it executes $\mathit{Look}$-$\mathit{Compute}$-$\mathit{Move}$ operations:

%\vspace{-3mm}
\begin{itemize}
\item {\bf Look:} The robot $r_i$ activates its sensors to obtain a snapshot  which consists of a pair of light and position for every robot with respect to the coordinate system of $r_i$. 
Since the result of this operation is a snapshot of the positions of all robots, the robot does not notice the movement, even if it sees other moving robots.
%Let ${\cal SS}_i(t)$ denote the snapshot of $r_i$ at time $t$. 
We assume that robots can observe all other robots~(unlimited visibility).
\item {\bf Compute:} The robot $r_i$ executes its algorithm using the snapshot and the color of its own light (if allowed by the model) and returns a destination point $des_i$ expressed in its coordinate system and a light $\ell_i \in L$ to which its own color is set.
\item {\bf Move:} The robot $r_i$ moves to the computed destination $des_i$.
A robot $r$ is said to {\em collide} with robot $s$ at time $t$ if $p(r,t)=p(s,t)$ and at time $t$ $r$ is performing $\mathit{Move}$. The collision is {\em accidental} if $r$'s destination is not $p(r,t)$.
Since robots are seen as points, we assume that accidental collisions are immaterial. A moving robot,
upon causing an accidental collision, proceeds in its movement without changes, in a ``hit-and-run'' fashion \cite{FPS}.
The robot may be stopped by an adversary before reaching the computed destination. If stopped before reaching its destination, a robot moves at least a minimum distance $\delta >0$. Note that without this assumption an adversary could make it impossible for any robot to ever reach its destination. 
If the distance to the destination is at most $\delta$, the robot can reach it. In this case, the movement is called {\em Non-Rigid}. Otherwise,
it is called {\em Rigid}. If the movement is Non-Rigid and robots know the value of $\delta$, it is called {\em Non-Rigid(+$\delta)$}. 
\end{itemize}

A scheduler decides which subset of  robots is activated for every configuration. 
The schedulers we consider are asynchronous and semi-synchronous and it is assumed that schedulers are {\em fair}, each robot is activated infinitely often.
\begin{itemize}
\item {\bf ASYNC:} 
The asynchronous (ASYNC) scheduler, activates the robots independently, and the duration of each $\mathit{Compute}$, $\mathit{Move}$ and between successive activities is finite and unpredictable. As a result, robots can be seen while moving and the snapshot and its actual configuration are not the same and so its computation may be done with the old configuration.
\item {\bf SSYNC:} 
The  semi-synchronous(SSYNC) scheduler activates a subset of all robots synchronously  and their $\mathit{Look}$-$\mathit{Compute}$-$\mathit{Move}$ cycles are performed at the same time. We can assume that activated robots at the same time obtain the same snapshot and their $\mathit{Compute}$ and $\mathit{Move}$ are executed instantaneously.
In SSYNC, we can assume that each activation defines discrete time called {\em round} and $\mathit{Look}$-$\mathit{Compute}$-$\mathit{Move}$ is performed instantaneously  in one round. 
\end{itemize}

As a special case of SSYNC, if all robots are activated in each round,
the scheduler is called full-synchronous (FSYNC).

In this paper, we are concerned with ASYNC and we assume the followings;
In a $\mathit{Look}$ operation, a snapshot of the environment at time $t_L$ is taken and we say that the {\em $\mathit{Look}$ operation is performed at time $t_L$.}
Each $\mathit{Compute}$ operation of $r_i$ is assumed to be done at time $t_C$ and the color of its light $\ell_i(t)$ and its pending destination $des_i$ are both set to the computed values for any time greater than $t_C$\footnote{Note that if some robot performs a $\mathit{Look}$ operation at time $t_C$, then
it observes the former color and if it does at time $t_C+\epsilon (\forall\epsilon>0)$, then
it observes the newly computed color.}.   
In a $\mathit{Move}$ operation, when the movement begins at time $t_B$ and ends at $t_E$, we say that it is performed during interval $[t_B, t_E]$, and the beginning (resp. ending) of the movement is denoted by $\mathit{Move_{BEGIN}}$ (resp. $\mathit{Move_{END}}$) occurring at time $t_B$ (resp. $t_E$).  
In the following, $\mathit{Compute}$, $\mathit{Move_{BEGIN}}$ and $\mathit{Move_{END}}$ are abbreviated as $\mathit{Comp}$, $\mathit{M_{B}}$ and $\mathit{M_{E}}$, respectively.
When a cycle has no actual movement (i.e., robots only change color and their destinations are the current positions),
we can equivalently assume that the $\mathit{Move}$ operation in this cycle is omitted, since we can consider 
the $\mathit{Move}$ operation to be performed just before the next $\mathit{Look}$ operation.

Without loss of generality, we assume the set of time instants at which the robots start executions of $\mathit{Look}$, $\mathit{Comp}$, $\mathit{M_B}$ and $\mathit{M_E}$ is $\N$.

We also consider the following restricted classes of ASYNC. 
Let a robot $r$ execute a cycle.
If no other robot can execute a $\mathit{Look}$ operation between the $\mathit{Look}$ operation of $r$ and its subsequent $\mathit{Compute}$ in that cycle,
the model is said to be {\em $\mathit{LC}$-atomic}. Thus we can assume that in the $\mathit{LC}$-atomic ASYNC model,
$\mathit{Look}$ and $\mathit{Comp}$ operations in every cycle are performed simultaneously (or atomically), say at time $t_{LC}$, and we say that the $\mathit{LC}$-operation is performed at time $t_{LC}$.

Similarly, if no other robot can execute a $\mathit{Look}$ operation between the operation $\mathit{M_B}$ of $r$ and its corresponding $\mathit{M_E}$,
the model is said to be {\em $\mathit{Move}$-atomic}.
In this case $\mathit{Move}$ operations in all cycles can be considered to be performed instantaneously and at time $t_M$.
In $\mathit{Move}$-atomic ASYNC, when a robot $r$ observes another robot  $r'$ performing a $\mathit{Move}$ operation at time $t_M$,
$r$ observes the snapshot after the moving of $r'$.

Since each operation occurs at integer times, when $\mathit{LC}$-operation is performed at time $t$ in $\mathit{LC}$-atomic ASYNC, we can assume  that the snapshot at $t$ is obtained at $t$ and the computation completes at $t+1$. Also when $\mathit{Move}$-operation begins ($\mathit{M_B}$ occurs) at time $t$ in $\mathit{Move}$-atomic ASYNC, $\mathit{M_E}$ can be assumed to occur at time $t+1$. Thus, if a robot $r$ observes another robot  $r'$ performing a $\mathit{Move}$ operation at time $t_M$,
then $r$ observes the snapshot before the moving of $r'$ until and at time $t$, and the snapshot after the moving of $r'$ from $t_M+1$.

In our settings, robots have  persistent lights and can change their colors instantly at each $\mathit{Compute}$ operation. 
We consider the following three robot models according to the visibility of lights.
\begin{itemize}
\item {\em full-light},
	a robot can observe the lights of other robots as well as its own, and it can also change the color of its own light.
\item {\em external-light}, 
	a robot can observe the light of other robots but not its own. It can however change the color of its own light in a ``write-only'' manner.
\item {\em internal-light},
	a robot can observe and change the color of its own light, but cannot observe the lights of other robots.
\end{itemize}

%When a robot performs Look operation in the internal-light model,  its snapshot is the same as that in the case of robots without lights. On the other hand, in %the full-light or external-light model,  it obtains a snapshot with locations of other robots and their colors. In this case we consider several types of snapshots %according to views robots observe.    
%vew of snapshot
%\noindent
%How to recognize lights of other robots

 \subsection{Rendezvous and \ML-Algorithms} 

An $n$-{\em Gathering} problem is defined as follows: given $n (\geq 2)$ robots initially  placed at arbitrary positions  in $\R^2$,
they congregate in finite time at a single location which is not predefined.
%\end{definition}
In the following, we consider the case where $n=2$ and
the $2$-Gathering problem is called {\em Rendezvous}.
%and the $n$-Gathering  problem ($n \geq 3)$ is simply called Gathering.
%Gathering is said to be {\em distinct} if all robots are initially placed in different positions.  An algorithm solving Gathering is said to be {\em self-stabilizing} if robots are initially set their lights to  arbitrary colors. 

When we consider algorithms on robots with lights, we exclude algorithms that solve Rendezvous only starting from initial settings in which robots have different colors of lights. That is, we consider Rendezvous algorithms that can solve Rendezvous even from initial settings in which all robots have  the  same color. %assume that all algorithms to solve Rendezvous all robots have the same colors of lights as the initial setting.
An algorithm solving Rendezvous is said to be {\em quasi-self-stabilizing} if it assumes that both robots always start with the same initial color chosen arbitrarily,
and it is {\em self-stabilizing} if the robots can start from an arbitrary color.

A particular class of algorithms, denoted by \ML, requires that robots only compute a destination point of the form $(1-\lambda)\cdot \mbox{me.position} + \lambda \cdot \mbox{other.position}$ for some $\lambda \in \R$, obtained as a function having only the colors as input (i.e., color of the other robot in the external-light) \cite{V}.
We call an algorithm in this class an \ML-algorithm.

 %Self-Stabilizing 
 %Definition and properties of \ML-algorithm

\section{Previous Results for Rendezvous}
\label{sec:PRR}
%Rendezvous is stated in Introduction??
%First we state previous results for Rendezvous without lights.
Rendezvous is trivially solvable in FSYNC but is not in SSYNC in general.
%The multiplicity detection does not help to solve Rendezvous and
%it is generally unsolved with SSYNC, even if chirality is assumed. 

\begin{theorem} \label{theorem:Rand_Impo}
{\em \cite{FPS}} Rendezvous is deterministically unsolvable in SSYNC even if chirality is assumed.
\end{theorem}

If robots have a constant number of colors in their lights, Rendezvous can be solved as shown in the following theorem (or Table~\ref{tab:Table-Rendezvous}).

\begin{theorem}\label{PRforSSLA}
%{\em \cite{FSVY,DFPSY,V}}
Rendezvous is solved by self-stabilizing \ML-algorithms
under the following assumptions;
\begin{enumerate}
%\item Full-light with $2$ colors, Non-Rigid and SSYNC,
\item full-light with $2$ colors, Non-Rigid and $\mathit{LC}$-atomic ASYNC {\em \cite{OWK}},
\item full-light with $3$ colors, Non-Rigid and ASYNC {\em \cite{V}},
\item external-light with $3$ colors, Non-Rigid and SSYNC {\em \cite{FSVY}}.
\end{enumerate} 
\end{theorem}

\begin{theorem}\label{PRforSSnonLA}{\em \cite{HDT}}
%{\em \cite{FSVY,DFPSY,V}}
Rendezvous is solved by a self-stabilizing non-\ML-algorithm
in full-light with $2$ colors, Non-Rigid and ASYNC.
%\begin{enumerate}
%\item Full-light with $2$ colors, Non-Rigid and SSYNC,
%\item Full-light with $2$ colors, Non-Rigid and ASYNC,
%\item Full-light with $3$ colors, Non-Rigid and ASYNC,
%\item external-light with $3$ colors, Non-Rigid and SSYNC\cite{FSVY}.
%\end{enumerate} 
\end{theorem}

\begin{theorem}\label{PRfornonSSnonLA}
%{\em \cite{FSVY,DFPSY,V}}
Rendezvous is solved by non-quasi-self-stabilizing non-\ML-algorithms
under the following assumptions;
\begin{enumerate}
\item full-light with $2$ colors, Non-Rigid and ASYNC {\em \cite{HDT}},
\item full-light with $2$ colors, Non-Rigid(+$\delta$) and ASYNC {\em \cite{OWK}},
\item external-light with $3$ colors, Non-Rigid(+$\delta$) and ASYNC {\em \cite{FSVY}},
\item external-light with $12$ colors, Rigid and ASYNC {\em \cite{FSVY}},
\item internal-light with $3$ colors, Non-Rigid(+$\delta$) and SSYNC {\em \cite{FSVY}},
\item internal-light with $6$ colors, Rigid and SSYNC {\em \cite{FSVY}}.
%\item[(1)] Rendezvous is solved in full-light, Non-Rigid and SSYNC with $2$ colors.
%\item[(2)] Rendezvous is solved in external-light, Non-Rigid and SSYNC with $3$ colors.
%\item[(3)] Rendezvous is solved in internal-light, Rigid and SSYNC with $6$ colors.
%\item[(4)]  Rendezvous is solved in internal-light, Non-Rigid(+$\delta$) and SSYNC with $3$ colors.% and knowledge of a minimum distance $\delta$ robots move.
%\item[(5)] Rendezvous is solved in full-light, Non-Rigid and ASYNC with $4$ colors.
%\item[(6)] Rendezvous is solved in external-light, Rigid and ASYNC with $12$ colors.
%\item[(7)] Rendezvous is solved in external-light, non-rigid(+$\delta$) and ASYNC with $3$ colors.% and knowledge of a minimum distance $\delta$ robots move.
\end{enumerate} 
\end{theorem}
%Figure is better!
%full-light > external-light > internal-light
%non-rigid with delta >rigid >non-rigid

%Impossibility and/or possibility results for Gathering are stated in the following theorems. 

Impossibility of Rendezvous \ML-algorithms is stated as follows.
 
\begin{theorem}\label{LBLalgorithms}
%{\em \cite{V,FSVY}}
\begin{enumerate}
\item In ASYNC and Rigid, Rendezvous is not solvable by any quasi-self-stabilizing \ML-algorithm with full-light of $2$ colors {\em \cite{V}}.
\item In ASYNC and Non-Rigid, Rendezvous is not solvable by any \ML-algorithm with full-light of $2$ colors {\em \cite{V}}.
\item In $\mathit{Move}$-atomic but non-$\mathit{LC}$-atomic ASYNC and Rigid, Rendezvous is not solvable by any \ML-algorithm with external-light of any number of colors {\em \cite{FSVY}}.%Rigid?
\item In SSYNC and Rigid, Rendezvous is not solvable by any \ML-algorithm with external-light of any number of colors {\em \cite{FSVY}}.%Rigid?
\end{enumerate}%check Rigid and non-LC-atomic!!!!
\end{theorem}

%Th.2(1): LC-atomic is a maximal subset of ASYNC such that with 2 colors
%Th.2(2): optimal with respect to the number of colors
In the following sections, we consider \ML-algorithms to solve Rendezvous on robots with external-lights and clarify the relationship among synchrony, the number of colors, movement restriction, and initial configurations.

%Rev_3 from here
\section{Rendezvous \ML-Algorithms for Robots with Three Colors of External Lights}
\label{sec:SSYNCRendezvousAlgorithms}

%\subsection{Properties for \ML-Algorithms}
\Newcodeline
\begin{algorithm}[h]
\caption{SS-Rendezvous-with-3-colors (scheduler, movement, initial-color)\cite{FSVY}}
\label{algo:Ren3}
{\footnotesize
%{\small 
\begin{tabbing}
111 \= 11 \= 11 \= 11 \= 11 \= 11 \= 11 \= \kill
{\em Parameters}: scheduler, movement-restriction, initial-color \crm
{\em Assumptions}: external-light, three colors ($A$, $B$ and $C$) \crm
\crm
\Cl \> {\bf case} other.light   {\bf of } \crm

\Cl \>$A$: \crm
\Cl \>\>$me.light \leftarrow B$ \crm
\Cl \>\>$me.des \leftarrow$ the midpoint of $me.position$ and $other.position$\crm
\Cl \> $B$: \crm
\Cl \> \>$me.light \leftarrow C$\crm
%\Cl \> \>$me.des \leftarrow me.position$ //stay\crm
\Cl \> $C$: \crm
\Cl \>\>$me.light \leftarrow A$ \crm 
\Cl \>\>$me.des \leftarrow other.position$ \crm 

\Cl \> {\bf endcase} 
\end{tabbing}
%}
}
\end{algorithm}

In what follows, two robots are denoted as $r$ and $s$.
Let $t_0$ be the starting time of the algorithm. 
Given a robot $robot$, an operation $op$($\in \{ \mathit{Look, Comp, LC, M_{B}, M_{E}} \}$), and a time $t$,
$t^+(robot, op)$ denotes the time $robot$ performs the first $op$ after $t$ (inclusive) if there exists such operation, and
$t^-(robot, op)$ denotes the time $robot$ performs the first $op$ before $t$ (inclusive)  if there exists such operation.
If $t$ is the time the algorithm terminates, $t^+(robot, op)$ is not defined for any $op$.
When $robot$ does not perform $op$ before $t$ and $t^-(robot, op)$ does not exist, $t^-(robot, op)$ is defined to be $t_0$.

A time $t_c$ is called a {\em cycle start time} ({\em cs-time}, for short), if the next performed operations of both $r$ and $s$ after $t$ are both $\mathit{Look}$, or
otherwise, the robots performing the operations neither change their colors of lights nor move.
In the latter case, we can consider that these operations can be performed before $t_c$ and the subsequent $Look$ operation can be performed as the first operation after $t_c$.

In \cite{FSVY}, a Rendezvous algorithm is shown in SSYNC and Non-Rigid with external-light of three colors (Algorithm~\ref{algo:Ren3}).

\begin{theorem}
{\em \cite{FSVY}}
Rendezvous is solved by SS-Rendezvous-with-3-colors(SSYNC, Non-Rigid, any).
It is a self-stabilizing \ML-algorithm.
\end{theorem}
%fromhere2
We will show that Algorithm~\ref{algo:Ren3} does not work in $\mathit{LC}$-atomic and $\mathit{Move}$-atomic ASYNC and Rigid, 
starting from the initial color $A$. 
In fact, in the next section, more generally
we will show that 
there exists
no \ML-algorithm to solve Rendezvous
in $\mathit{LC}$-atomic and $\mathit{Move}$-atomic ASYNC and Non-Rigid with three colors of external-lights. We also show that there exists no quasi-self-stabilizing \ML-algorithm to solve Rendezvous if we change the assumption of Non-Rigid to Rigid. %in $\mathit{LC}$-atomic and $Move$-atomic ASYNC and Rigid, if it uses three colors of external-lights. 
However, we show that
there exists a non-quasi-self-stabilizing \ML-algorithm to solve Rendezvous in $\mathit{LC}$-atomic ASYNC and Rigid with three colors of external-lights.

%When two robots perform $\mathit{LC}$-operations simultaneously,
%note in the following.
%Let $\ell(r,t)=\alpha$ and $\ell(s,t)=\beta$.
%Suppose $r$ and $s$ change their colors into $\alpha'$ and $\beta'$,
%if $r$ and $s$ observe colors $\beta$ and $\alpha$, respectively.
%When two robots $r$ and $s$ perform $\mathit{LC}$-operation at the same time $t'=t^+(r,LC)=t^+(s,LC)$, we assume that $r$ and $s$ observe $s$ with $\beta$ and $r$ with $\alpha$ at time $t'$, and $r$ with $\alpha'$ and $s$ with $\beta'$ are observed at any time $t''(>t')$, respectively. 

\Newcodeline
\begin{algorithm}[h]
\caption{NonQSS-Rendezvous-with-3-colors (scheduler, movement, initial-color)}
\label{algo:NonQSSRen3}
{\footnotesize
%{\small 
\begin{tabbing}
111 \= 11 \= 11 \= 11 \= 11 \= 11 \= 11 \= \kill
{\em Parameters}: scheduler, movement-restriction, initial-color \crm
{\em Assumptions}: external-light, three colors ($A$, $B$ and $C$) \crm
\crm
\Cl \> {\bf case} other.light   {\bf of } \crm

\Cl \>$A$: \crm
\Cl \>\>$me.light \leftarrow B$ \crm
\Cl \>\>$me.des \leftarrow$ the midpoint of $me.position$ and $other.position$\crm
\Cl \> $B$: \crm
\Cl \> \>$me.light \leftarrow C$\crm

\Cl \>$C$: \crm
\Cl \>\>$me.light \leftarrow B$ \crm 
\Cl \> \>$me.des \leftarrow other.position$ \crm
%\Cl \>\>\>$me.des \leftarrow me.position$ // stay\crm 

\Cl \> {\bf endcase} 
\end{tabbing}
%}
}
\end{algorithm} 

\begin{theorem}
Rendezvous is solved by NonQSS-Rendezvous-with-3-colors($\mathit{LC}$-atomic ASYNC, Rigid, A).
It is a non-quasi-self-stabilizing \ML-algorithm.
\end{theorem}
\begin{proof}
%Let $t_0$ be the starting time and 
Let $\ell(r,t_0)=\ell(s,t_0)=A$.
There are two cases: (I)~$r$ and $s$ perform $\mathit{LC}$-operations at the same time, and (II)~one
robot, say $r$, performs its $\mathit{LC}$-operation earlier than the other robot~($s$).

(I)~Let $t_1=t^+_0(r,LC)=t^+_0(s,LC)$. We consider the ends of these cycles for both robots. They move to the midpoint in their $\mathit{Move}$-operations and we can assume that $t^+_1(s,\mathit{M_E})(=t_2) \leq t^+_1(r, \mathit{M_E})(=t_3)$. 
If $s$ does not perform any cycle between $t_2$ and $t_3$, the distance of $r$ and $s$ becomes $0$ at $t_3$ and
$\ell(r,t_3)=\ell(s,t_3)=B$.
Otherwise, note that  $\ell(r,t_1+1)=B$ and
consider that $s$ performs $\mathit{LC}$ operations between $t_2+1$ and $t_3$.
If $s$ performs the first $\mathit{LC}$ at $t (t_2 < t \leq t_3)$, then $s$ changes its color into $C$ at $t+1$ but does not move after $t+1$, and
$s$ retains the color $C$ after the $\mathit{LC}$ and does not move until $t_f=max(t'+1,t_3)$, where $t'$ is the time of the last $\mathit{LC}$ $s$ performs.  Thus, the distance of the two robots becomes $0$ at time $t_f$ and $\ell(r,t_f)=B$ and $\ell(s,t_f)=C$.
It can be verified that they do not move after $t_f$ in either cases.

(II)~Let $t_1=t^+_0(r,\mathit{LC}) < t_2=t^+_0(s,\mathit{LC})$ and let $t_3=t^+_1(r,\mathit{LC})$.
If $(t_1 <) t_2 <t_3$, $r$ has moved to the midpoint of $p(r,t_0)$ and $p(s,t_0)$ and its color is $B$ and $s$ stays at $p(s,t_0)$ and its color is $C$ at time $t_3$. Then $r$ observes the color $C$ of $s$ at $t_3$ and moves to the position of $s$ and the color of $r$ is $B$ after $t_3$. Since $s$ retains the color ($C$) and stays at the same position even if it performs cycles after $t_3$, the distance of $r$ and $s$ becomes $0$ at $t^+_2(r,\mathit{M_E})$ and 
they do not move after $t^+_2(r,\mathit{M_E})$.

If $t_2 \leq t_3$, assume that $r$ performs $k (\geq 1)$ cycles before $t_3$ and the last $\mathit{LC}$-operation is performed at $t_k$.
Then $r$ repeats $k-1$ moves to the midpoints and its color is $B$ at time $t_k$.
Since $t_k \leq t_2$, $s$ observes the color $B$ of $r$ and its color is $C$ and it stays at $p(s,t_0)$ after $t_2$.
Since $r$ observes the color $C$ at $t^+_k(r,\mathit{LC})$ and moves to the position $p(s,t_0)$, the distance of $r$ and $s$ becomes $0$ 
at $t^+_k(r,\mathit{M_E})$. 

If the initial colors of $r$ and $s$ are $B$ or $C$, they can repeatedly swap their positions.
Thus Algorithm~\ref{algo:NonQSSRen3} is a non-quasi-self-stabilizing \ML-algorithm.
 \end{proof}

%Refer previous result for SSYNC, external, Non-Rigid Rendezvous \ML-algorithm, we show the number of lights is optimal\footnote{Although Giovanni stated it, there is no proof for it.} 
%ALgorithm1 and 2 are very interesting!!
In the following,
we derive lower bounds on the number of colors of external-lights.  
In order to do so, we introduce some notation for \ML-algorithms and show their properties.

In \ML-algorithms, the next color and destination (denoted by $\lambda$) is determined only by 
the current color observed by the robot. Thus an \ML-algorithm is represented by an edge-labeled directed graph $G_{\mathcal L} = (V_{\mathcal L},E_{\mathcal L},\ell_{\mathcal L})$, where $V_{\mathcal L}$ is a set of colors used in the algorithm, $E_{\mathcal L}$ is a set of transitions from current colors observed by the robots to the next colors computed by the robots, and 
$\ell_{\mathcal L}$ is an edge-labeled function from $E_{\mathcal L}$ to \R.
Edge $e=(c_1,c_2) \in E_{\mathcal L}$ and $\ell_{\mathcal L}(e)=\lambda$ mean that
when a robot observes color $c_1$ of the other robot, it changes its color to $c_2$ and moves to the point decided by the value $\lambda$. Also the out-degree of each node must be one, since we consider deterministic \ML-algorithms. Thus, when the number of nodes in $G_{\mathcal L}$ is $k$, $G _{\mathcal L}$ has $k$ edges.
For example, Algorithms~\ref{algo:Ren3} and \ref{algo:NonQSSRen3} are represented by the following directed graphs  $G_{\mathcal L1}$ and $G_{\mathcal L2}$, respectively.
 
$G_{\mathcal L1} = (V_{\mathcal L1},E_{\mathcal L1},\ell_{\mathcal L1})$,
where $V_{\mathcal L1}=\{ A,B,C\}$, $E_{\mathcal L1}=\{(A,B), (B,C), (C,A) \}$ and
$\ell_{\mathcal L1}((A,B))=1/2$, $\ell_{\mathcal L1}((B,C))=0$ and $\ell_{\mathcal L1}((C,A))=1$ (Fig.~\ref{Graphforalg1and2}(a)).

$G_{\mathcal L2} = (V_{\mathcal L2},E_{\mathcal L2},\ell_{\mathcal L2})$,
where $V_{\mathcal L2}=\{ A,B,C\}$, $E_{\mathcal L2}=\{(A,B), (B,C), (C,B) \}$ and
$\ell_{\mathcal L2}((A,B))=1/2$, $\ell_{\mathcal L1}((B,C))=0$ and $\ell_{\mathcal L2}((C,B))=1$ (Fig.~\ref{Graphforalg1and2}(b)).

 \begin{figure*}%{R}{0.5\textwidth} 
%\vspace{3cm}
 \begin{center}
    \includegraphics[scale=0.5]{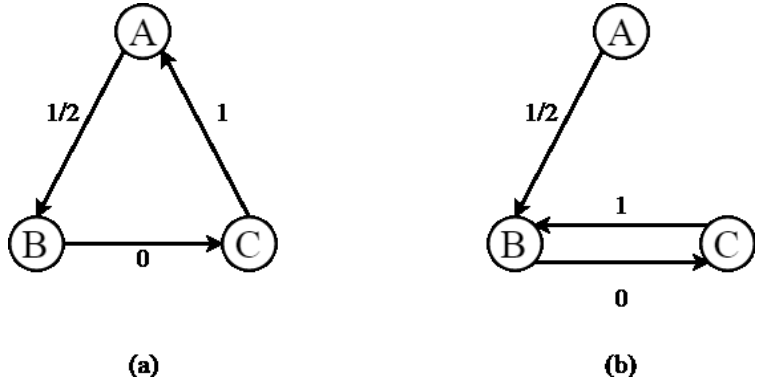}
%    \includegraphics[scale=0.8]{Lemma7(b)_v2.eps}
%    %\includegraphics[scale=0.5]{Fig1.eps}
    \caption{Graph  representations for Algorithms~\ref{algo:Ren3} (a) and  \ref{algo:NonQSSRen3} (b).}
    \label{Graphforalg1and2}
  \end{center}
%  \vspace{-20pt}
%  \vspace{1pt}
\end{figure*} 

In what follows, we identify an \ML-algorithm with  its graph representation and $e=(c_1,c_2) \in E_{\mathcal L}$ and $\ell_{\mathcal L}(e)=\lambda$ are denoted as $c_1 \EL c_2$.

%% XD: HERE

\begin{lemma}\label{PofL}
Let $A_{\mathcal L}$ be an \ML-algorithm solving Rendezvous in SSYNC and Rigid with external-light. If $A_{\mathcal L}$ starts from an initial settings such that both robots have the same color, then 
$A_{\mathcal L}$ has the following properties.
\begin{enumerate}
\item There is a color $X$ such that $A_{\mathcal L}$ must have an edge  $X \EH Y$.
\item There is a color $X$ such that $A_{\mathcal L}$ must have an edge  $X \EO Y$.
\item There is a color $X$ such that $A_{\mathcal L}$ must have an edge  $X \EZ Y$.

%\item $A_{\mathcal L}$ must not have any edge $X \EL X$ for any $\lambda$.
%\item $A_{\mathcal L}$ must not have both edges $X \EL Y$ and $Y \EL X$ ($X \neq Y$) for any $\lambda$.
\end{enumerate}

\end{lemma}
\begin{proof}~
	\begin{enumerate}
	\item Assume that $r$ and $s$ become active in each round (FSYNC). If there exists no edge $X \EH Y$, then they never attain Rendezvous.

	\item Assume that $r$ and $s$ become active alternately. If there exists no edge $X \EO Y$, then neither robot can ever reach the other.

	\item Assume that $r$ and $s$ become active alternately. Once a robot will follow the edge $X \EO Y$ in some round, then let both robots be activated in that round. Alternatively, if there exists no edge $X \EZ Y$, then they fail to Rendezvous in the round.
		\QED
	\end{enumerate}
\end{proof}
%\begin{proof}
%
%1. Assume that $r$ and $s$ become active in each round (FSYNC). If there does not exist the edge $X \EH Y$, they never attain Rendezvous.
%
%2. Assume that $r$ and $s$ become active alternately. If there does not exist the edge $X \EO Y$, each robot cannot reach the other robot forever.
%
%3. Assume that $r$ and $s$ become active alternately. When a robot performs the edge $X \EO Y$ in some round, we make the two robots active in this round. If there does not exist the edge $X \EZ Y$, they fail to Rendezvous in this round.
%
%\end{proof}

Lemma~\ref{PofL} implies that any \ML-algorithm must contain three different edges beginning with different colors.

\begin{theorem}
In any Rendezvous \ML-algorithm with external-light, robots must have three colors in SSYNC and Rigid.
\end{theorem}

This theorem implies that Algorithm~\ref{algo:Ren3} has the optimal number of colors of external-lights in SSYNC. Note that it is self-stabilizing and works in Non-Rigid.
On the other hand,
if we assume Rigid movement, we can show the \ML-algorithm with three colors to solve Rendezvous in $\mathit{LC}$-atomic ASYNC, which is however not quasi-self-stabilizing.
In the next section, 
we will show a quasi-self-stabilizing \ML-algorithm with four colors and  a self-stabilizing \ML-algorithm with five colors to solve Rendezvous in $\mathit{LC}$-atomic ASYNC and Non-Rigid. We will also show that the number of colors used in each algorithm is optimal.
%%%checked here

\section{Optimal Rendezvous \ML-Algorithms for $\mathit{LC}$-atomic ASYNC Robots with External Lights}\label{ASYNCRendezvousAlgorithms} 

\subsection{Lower Bounds}

In this subsection we first show that there exists no (not even quasi-self-stabilizing) Rendezvous \ML-algorithm
 with external light of $3$ colors in $\mathit{LC}$-atomic and $Move$-atomic ASYNC in Non-Rigid. %and SS...

If there exists such an \ML-algorithm, the algorithm must be an edge-labeled directed graph
$G_{\mathcal L} = (V_{\mathcal L},E_{\mathcal L},\ell_{\mathcal L})$ such that
$V_{\mathcal L}=\{A,B,C\}$(three colors) and  $\ell_{\mathcal L}(E_{\mathcal L})=\{0,1/2,1\}$ (by Lemma\ref{PofL}) and one of the following edge sets:

\begin{enumerate}
\item $E_{\mathcal L}$ contains a self-loop edge, say $(A,A)$,
\item $E_{\mathcal L}$ contains  both directed edges, say $(B,C)$ and $(C,B)$, or
\item $E_{\mathcal L} = \{(A,B), (B,C), (C,A)\}$.
\end{enumerate} 

For Case~1. If the algorithm does not contain both directed edges, it can be verified that no algorithm can solve Rendezvous in SSYNC and Rigid. That is, if the algorithm starts with a color consisting of a self-loop edge, then it cannot solve Rendezvous since it cannot use more than one color. If the algorithm starts with a color not consisting of a self-loop edge, the color of both robots can be changed into the color with the self-loop edge without attaining Rendezvous. Thus, the algorithm also fails to Rendezvous
in this case.

For Case~2. If algorithms do not contain self-loop edges, their graphs are the same as 
that of Algorithm~\ref{algo:NonQSSRen3}.
But it can be verified that Algorithm~\ref{algo:NonQSSRen3} fails to solve Rendezvous
in SSYNC, Rigid and starting from color $B$ or $C$, or SSYNC, Non-Rigid and starting from any color. It is easily verified that other algorithms with different edge-labeled functions 
fail to solve Rendezvous in SSYNC and Rigid starting from any color.
If algorithms contain self-loop edges (both directed edges and a self-loop edge),
since they can use only less than three colors even if starting from any color,
they never solve Rendezvous in SSYNC and Rigid.
 
In Case~3, there are essentially two algorithms.%labeLの付け替えを除いて

\begin{enumerate}
\item[(a)] $\ell_{\mathcal L}((A,B))=1/2$, $\ell_{\mathcal L}((B,C))=0$, and $\ell_{\mathcal L}((C,A))=1$ (denoted as Alg-(a)),
\item[(b)] $\ell_{\mathcal L}((A,B))=1/2$, $\ell_{\mathcal L}((B,C))=1$, and $\ell_{\mathcal L}((C,A))=0$ (denoted as Alg-(b)).
\end{enumerate}

Note that Alg-(a) is Algorithm~\ref{algo:Ren3}.

%schedulers notations for Impossibility
We introduce special schedules to analyze \ML-algorithms
to solve Rendezvous in $\mathit{LC}$-atomic ASYNC, with which we show that these algorithms do not work well.

Let $([\alpha_1,\beta_1], [\alpha_2,\beta_2], \ldots)$ be a sequence of operations that robots~$r$ and~$s$ perform, where $r$ and $s$ perform $\alpha_i$ and $\beta_i$ at time $t_i$ ($1 \leq i$), respectively,
and $\alpha_i$ and $\beta_i$ are taken from $\mathit{LC}$-operation (denoted as $\mathit{LC}$),
$Move$-operations, $Move_B$, $M_E$ or $M$ (if $\mathit{Move}$-atomic) (denoted as $M$), and a ``\emph{no-op}'' operation (denoted as $-$). For example, $([\mathit{LC},-],[-,\mathit{LC}],[\mathit{M},-],[-,\mathit{M}])$ denotes that $r$ performs $\mathit{LC}$ and $\mathit{M}$ at times $t_1$ and $t_3$ and $s$ performs $\mathit{LC}$ and $\mathit{M}$ at times $t_2$ and $t_4$, which is in $\mathit{LC}$-atomic $\mathit{Move}$-atomic ASYNC. Similarly, $(\mathit{[LC,LC],[M,M]})$ denotes that $r$ and $s$ perform $\mathit{LC}$ at time $t_1$ and perform $\mathit{M}$ at time $t_2$, which is in FSYNC.
The former is called alternate schedule and denoted as $alt$ and the latter is called simultaneous schedule and denoted as $sim$.

 \begin{figure*}%{R}{0.5\textwidth} 
%\vspace{3cm}
 \begin{center}
    \includegraphics[scale=0.5]{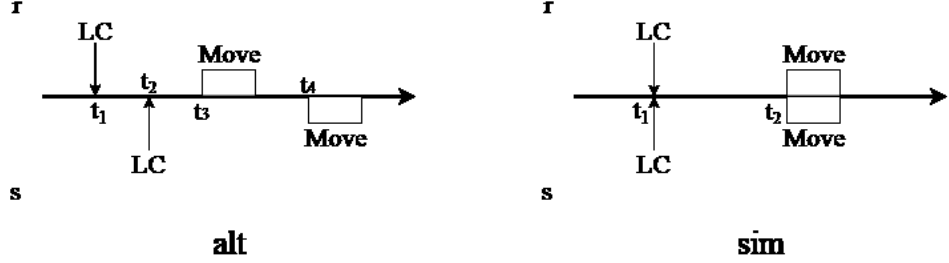}
%    \includegraphics[scale=0.8]{Lemma7(b)_v2.eps}
%    %\includegraphics[scale=0.5]{Fig1.eps}
    \caption{Special schedules $alt$ and $sim$.}
    \label{AltandSim}
  \end{center}
%  \vspace{-20pt}
%  \vspace{1pt}
\end{figure*} 

Assume that $r$ and$s$ have colors $c_r$ and $c_s$ at some time $t$ and let $d_t=dis(p(r,t),p(s,t))$. Let $(c_r,c_s;d_t)$ denote a configuration of a pair of colors of robots and its distance at $t$. When a configuration $(c_r,c_s;d_t)$ is changed into another one $(c'_r,c'_s;d_{t'})$ by performing an algorithm $alg$ with a schedule $sch$, we denote
$(c_r,c_s;d_t) \overset{sch}{\rightarrow} (c'_r,c'_s;d_{t'})_{alg}$, where $t'$ is the time after which the robots have performed $alg$ with the schedule $sch$.
The suffix $alg$ is usually omitted when the algorithm is apparent from the context.

We show that Alg-(a) and Alg-(b) cannot work from any initial configuration of the same color.

\begin{lemma}
Alg-(a) cannot solve Rendezvous in $\mathit{LC}$-atomic and $\mathit{Move}$-atomic ASYNC and Rigid.
\end{lemma}
\begin{proof}
Let $t_0$ be the starting time of Alg-(a) and let $d=dis(p(r,t_0),p(s,t_0))$.
Since $(B,B;d) \overset{sim}{\rightarrow} (C,C;d)$,
$(C,C;d) \overset{sim}{\rightarrow} (A,A;d)$, and 
$(A,A;d) \overset{alt}{\rightarrow} (B,C;d/2)$,
it is sufficient to show that Alg-(a) cannot solve Rendezvous from the initial 
configuration $(B,C;d)$. 
Since $(B,C;d) \overset{alt}{\rightarrow} (A,B;d/2)$, when $r$ performs one 
cycle from $(A,B;d/2)$, the configuration becomes $(C,B;d/4)$.
Thus, since this configuration repeats, Alg-(a) never achieve Rendezvous from the initial configuration $(B,C;d)$.

\end{proof}

The next lemma can be proved similarly.

\begin{lemma}
Alg-(b) cannot solve Rendezvous in $\mathit{LC}$-atomic and $\mathit{Move}$-atomic ASYNC and Rigid.
\end{lemma}

\begin{theorem}\label{Imp3}
There exists no (not even quasi-self-stabilizing) \ML-algorithm
of Rendezvous with external light of $3$ colors in $\mathit{LC}$-atomic and $\mathit{Move}$-atomic ASYNC and Non-Rigid. 
Furthermore, there exists no quasi-self-stabilizing \ML-algorithm
of Rendezvous with external light of $3$ colors in $\mathit{LC}$-atomic and $\mathit{Move}$-atomic ASYNC and Rigid.
%There does not exist any SS \ML-algorithm of Rendezvous with $3$ colors of external lights in $LC$-atomic and $Move$-atomic ASYNC and Rigid.
\end{theorem}

In an argument similar to the one above, we show that there exists no self-stabilizing  \ML-algorithm
of Rendezvous with external-light of  $4$ colors in $\mathit{LC}$-atomic and $\mathit{Move}$-atomic ASYNC and Rigid. 

%The number of scc's \geq 2 trivial!
If there exists such an \ML-algorithm, the algorithm must be an edge-labeled directed graph
$G_{\mathcal L} = (V_{\mathcal L},E_{\mathcal L},\ell_{\mathcal L})$ shch that
$V_{\mathcal L}=\{A,B,C,D\}$(four colors) and  $\ell_{\mathcal L}(E_{\mathcal L}) \supseteq \{0,1/2,1\}$ (by Lemma\ref{PofL}).
If the number of strongly connected components for $G_{\mathcal L}$ is at least two,  then 
there exists an initial configuration of both robots with a same color, from which an algorithm cannot use four colors, 
it cannot solve Rendezvous by Theorem~\ref{Imp3}. Then the remaining case is that these graphs have one strongly connected component  (one cycle) and have one of the following edge sets:

\begin{enumerate}
\item[(1)] $\ell_{\mathcal L}((A,B))=1/2$, $\ell_{\mathcal L}((B,C))=0$, $\ell_{\mathcal L}((C,D))=1,$ and $\ell_{\mathcal L}((D,A))=\lambda$(denoted as Alg-(1)),
\item[(2)] $\ell_{\mathcal L}((A,B))=1/2$, $\ell_{\mathcal L}((B,C))=1$, $\ell_{\mathcal L}((C,D))=0,$ and $\ell_{\mathcal L}((D,A))=\lambda$(denoted as Alg-(2)),
\item[(3)] $\ell_{\mathcal L}((A,B))=1/2$, $\ell_{\mathcal L}((B,C))=0$, $\ell_{\mathcal L}((C,D))=\lambda,$ and $\ell_{\mathcal L}((D,A))=1(\lambda \neq 1)$(denoted as Alg-(3)),
\item[(4)] $\ell_{\mathcal L}((A,B))=1/2$, $\ell_{\mathcal L}((B,C))=1$, $\ell_{\mathcal L}((C,D))=\lambda,$ and $\ell_{\mathcal L}((D,A))=0(\lambda \neq 0)$(denoted as Alg-(4)),
\item[(5)] $\ell_{\mathcal L}((A,B))=1/2$, $\ell_{\mathcal L}((B,C))=\lambda$, $\ell_{\mathcal L}((C,D))=0,$ and $\ell_{\mathcal L}((D,A))=1(\lambda \neq 1)$(denoted as Alg-(5)),
\item[(6)] $\ell_{\mathcal L}((A,B))=1/2$, $\ell_{\mathcal L}((B,C))=\lambda$, $\ell_{\mathcal L}((C,D))=1,$ and $\ell_{\mathcal L}((D,A))=0(\lambda \neq 0)$(denoted as Alg-(6)).
\end{enumerate}

\begin{lemma}\label{ImpAlg4}
Alg-(1)-Alg-(6) cannot solve Rendezvous in $\mathit{LC}$-atomic and $\mathit{Move}$-atomic ASYNC and Rigid from some initial configuration.
\end{lemma}
\begin{proof}
Let $t_0$ be the starting time of Alg-($i$)($1 \leq i \leq 6)$ and let$d=dis(p(r,t_0),p(s,t_0))$.
In each case, we show initial configurations and schedules which repeat forever and therefore cannot solve Rendezvous.
%If $\lambda=1$, the initial configuration is $\ell(r,t_0)=B$ and $\ell(s,t_0)=C$ and the schedule is that
%$(B,C;d)\overset{alt}{\rightarrow} (D,A;d)\overset{alt}{\rightarrow} (B,C;d/2)$.

(1) Algo-(1): If $\lambda\neq 0$, the initial configuration is $\ell(r,t_0)=D$ and $\ell(s,t_0)=A$ and the schedule is 
$(D,A;d)\overset{alt}{\rightarrow} (B,C;d/2)\overset{alt}{\rightarrow} (D,A;d\lambda/2)$.

If $\lambda =0$, the initial configuration is $\ell(r,t_0)=A$ and $\ell(s,t_0)=C$ and the schedule is 
$(A,C;d)\overset{sim}{\rightarrow} (D,B;d/2)\overset{sim}{\rightarrow}(C,A;d/2)$.

(2) Alg-(2): If $\lambda=1$, the initial configuration is $\ell(r,t_0)=A$ and $\ell(s,t_0)=B$ and the schedule is 
$(A,B;d)\overset{alt}{\rightarrow} (C,D;d/2)\overset{alt}{\rightarrow} (A,B;d/4)$.

If $\lambda \neq1$, the initial configuration is $\ell(r,t_0)=B$ and $\ell(s,t_0)=C$ and the schedule is 
$(B,C;d)\overset{alt}{\rightarrow} (D,A;d(1-\lambda))\overset{sim}{\rightarrow} (B,C;d(1-\lambda)/2)$.

(3) Alg-(3): Since $\lambda \neq1$, the initial configuration is $\ell(r,t_0)=A$ and $\ell(s,t_0)=B$ and the schedule is 
$(A,B;d)\overset{alt}{\rightarrow} (C,D;d(1-\lambda))\overset{alt}{\rightarrow} (A,B;d(1-\lambda^2)/2)$.

(4) Alg-(4): Since $\lambda \neq 0$, the initial configuration is $\ell(r,t_0)=A$ and $\ell(s,t_0)=B$ and the schedule is 
$(A,B;d)\overset{alt}{\rightarrow} (C,D;d(1-\lambda))\overset{alt}{\rightarrow} (A,B;d\lambda/2)$.

(5) Alg-(5): Since $\lambda \neq 1$, the initial configuration is $\ell(r,t_0)=A$ and $\ell(s,t_0)=B$ and the schedule is 
$(A,B;d)\overset{alt}{\rightarrow} (C,D;d\lambda)\overset{alt}{\rightarrow} (A,B;d(1-\lambda)/2)$.

(6) Alg-(6): Since $\lambda \neq 0$, the initial configuration is $\ell(r,t_0)=A$ and $\ell(s,t_0)=B$ and the schedule is 
$(A,B;d)\overset{alt}{\rightarrow} (C,D;d\lambda)\overset{alt}{\rightarrow} (A,B;d\lambda/2)$.

All algorithms except Alg-(1)($\lambda=0$) fail to achieve Rendezvous from any initial configuration in which both robots have the same color,
since all these initial configurations can be reached from any configuration with same colors.
Note that the initial configuration in the case Alg-(1)($\lambda=0$) cannot be reached from any configuration with same colors.
In fact,  we will show in the next subsection that Alg-(1)($\lambda=0$)  is a quasi-self-stabilizing \ML-algorithm with four colors in $\mathit{LC}$-atomic ASYNC and Non-Rigid.
\end{proof}

\begin{theorem}
There exists no self-stabilizing \ML-algorithm of Rendezvous with external-light of $4$ colors in $\mathit{LC}$-atomic and $\mathit{Move}$-atomic ASYNC and Rigid.
\end{theorem}

\subsection{Optimal \ML-algorithms}

In this subsection, we show two optimal \ML-algorithms of Rendezvous, one is quasi-self-stabilizing (Algorithm~\ref{algo:Ren4}) with $4$ colors and the other is self-stabilizing (Algorithm~\ref{algo:Ren5}) with $5$ colors.

Algorithm~\ref{algo:Ren4} (QSS-Rendezvous-with-4-colors ($\mathit{LC}$-atomic ASYNC, Non-Rigid, initial-light) satisfies the following lemmas.
Let $t_c$ be a cs-time of Algorithm~\ref{algo:Ren4}.
%cs-time: $t_c$
%Rendezvous algorithm
\Newcodeline
\begin{algorithm}[h]
\caption{QSS-Rendezvous-with-4-colors ($\mathit{LC}$-atomic ASYNC, Non-Rigid, initial-color)}
\label{algo:Ren4}
{\footnotesize
%{\small 
\begin{tabbing}
111 \= 11 \= 11 \= 11 \= 11 \= 11 \= 11 \= \kill
{\em Parameters}: scheduler, movement-restriction, initial-color \crm
{\em Assumptions}: external-light, four colors ($A$, $B$, $C$ and $D$) \crm

\Cl \> {\bf case} other.light   {\bf of } \crm

\Cl \> $A$: \crm
\Cl \> \> $me.light \leftarrow B$ \crm
\Cl \> \> $me.des \leftarrow$ the midpoint of $me.position$ and $other.position$\crm
\Cl \> $B$: \crm
\Cl \>\>$me.light \leftarrow C$\crm 
\Cl \> $C$: \crm
\Cl \> \>$me.light \leftarrow D$ \crm
\Cl \> \>$me.des \leftarrow other.position$\crm 
\Cl \> $D$: \crm
\Cl \> \>$me.light \leftarrow A$ \crm 

\Cl \> {\bf endcase} 
\end{tabbing}
%}
}
\end{algorithm}

%Rendezvous algorithm
\Newcodeline
\begin{algorithm}[h]
\caption{SS-Rendezvous-with-5-colors ($\mathit{LC}$-atomic ASYNC, Non-Rigid, initial-color)}
\label{algo:Ren5}
{\footnotesize
%{\small 
\begin{tabbing}
111 \= 11 \= 11 \= 11 \= 11 \= 11 \= 11 \= \kill
{\em Parameters}: scheduler, movement-restriction, Initial-color) \crm
{\em Assumptions}: external-light, five colors ($A$, $B$, $C$, $D$ and $E$) \crm

\Cl \> {\bf case} other.light   {\bf of } \crm

\Cl \> $A$: \crm
\Cl \> \> $me.light \leftarrow B$ \crm
\Cl \> \> $me.des \leftarrow$ the midpoint of $me.position$ and $other.position$\crm
\Cl \> $B$: \crm
\Cl \>\>$me.light \leftarrow C$\crm 
\Cl \> $C$: \crm
\Cl \> \>$me.light \leftarrow D$ \crm
\Cl \> \>$me.des \leftarrow other.position$\crm 
\Cl \> $D$: \crm
\Cl \> \>$me.light \leftarrow E$ \crm 
\Cl \> $E$: \crm
\Cl \> \>$me.light \leftarrow A$ \crm 
\Cl \> {\bf endcase} 
\end{tabbing}
%}
}
\end{algorithm}

 \begin{figure*}%{R}{0.5\textwidth} 
%\vspace{3cm}
 \begin{center}
   \includegraphics[scale=0.5]{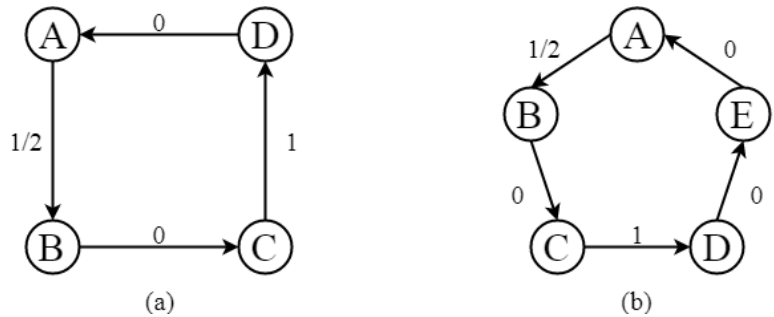}
%    \includegraphics[scale=0.8]{Lemma7(b)_v2.eps}
%    %\includegraphics[scale=0.5]{Fig1.eps}
    \caption{Graph  representations for Algorithms~\ref{algo:Ren4} (a) and  \ref{algo:Ren5} (b).}
    \label{Graphforalg3and4}
  \end{center}
%  \vspace{-20pt}
%  \vspace{1pt}
\end{figure*} 

\begin{lemma}\label{lemma-zero4}
If $dis(p(r,t_c),p(s,t_c))=0$ and Algorithm~\ref{algo:Ren4} is performed starting from $t_c$, $dis(p(r,t),p(s,t))=0$ for any $t \geq t_c$.
\end{lemma}
\begin{proof}
Since $dis(p(r,t_c),p(s,t_c))=0$, any move operation becomes no move (stay).
\end{proof}

\begin{lemma}\label{generalprop4}
Let $\alpha=B$ and $\beta=C$ or $\alpha=D$ and $\beta=A$.
If $\ell(r,t_c)=\alpha$ and $\ell(s,t_c)=\beta$ in Algorithm~\ref{algo:Ren4}, then  % such that $\alpha \rightarrow \beta$ and let $t_1=t^+_c(r,LC)$, then
$\ell(s,t)=\beta$  and $p(s,t)=p(s,t_c)$ for any $t(t_c \leq t \leq t_1=t^+_c(r,\mathit{LC}))$.
\end{lemma}
\begin{proof}
When $s$ with color $\beta$ observes $r$ with color $\alpha$ at $t(t_c \leq t \leq t_1$), $s$ does not change its color at time $t$ and stays at position $p(s,t_c)$.
\end{proof}

\begin{lemma}\label{lemma-BC}
If Algorithm~\ref{algo:Ren4} starts with $\{\ell(r,t_c), \ell(s,t_c)\} =\{B, C\}$,  %such that $\alpha \rightarrow \beta$,
for any schedule of two robots after $t_c$,
there is a cs-time $t^* (\geq t_c)$ such that $dis(p(r,t^*),p(s,t^*))=0$,
or $dis(p(r,t^*),p(s,t^*)) \leq dis(p(r,t_c),p(s,t_c))-\delta$ 
and $\{\ell(r,t^*), \ell(s,t^*)\} =\{C, D\}$ or $\{\ell(r,t^*), \ell(s,t^*)\} =\{D, A\}$ .
\end{lemma}
\begin{proof}
We can assume that $\ell(r,t_c)=B$ and $\ell(s,t_c)=C$.
Let $t_r=t^+_c(r,\mathit{LC})$. % and $t_s=t^+_c(s,LC)$.
By Lemma~\ref{generalprop4}, $\ell(s,t)=C$ and $s$ stays at $p(s,t_C)$ at time $t ( \leq t_r)$ even if $s$ performs several cycles before $t_r$.
When $r$ performs an $\mathit{LC}$-operation at time $t_r$, it changes its color to $D$ at $t_r+1$ and will move to position $p(s,t_c)$ and let $t'=t^+_r(r, \mathit{M_E})$.
Since $s$ stays at $p(s,t_c)$ until $t'$ whether $s$ is active ($s$ observes $B$ or $D$) or not ($s$ stays),
$t'+1$ becomes a cs-time $t^*$ satisfying the conditions of the lemma.
In fact, 
if $dis(p(r,t_c),p(s,t_c)) \leq \delta$, then $dis(p(r,t^*),p(s,t^*))=0$, otherwise,
$dis(p(r,t^*),p(s,t^*)) \leq dis(p(r,t_c),p(s,t_c))-\delta$.
If $s$ becomes active at $t (t_r \leq t \leq t')$, $s$ changes its color to $A$.
Otherwise, its color remains $C$. Thus,
$\{\ell(r,t^*), \ell(s,t^*)\} =\{C, D\}$ or $\{\ell(r,t^*), \ell(s,t^*)\} =\{D, A\}$. 
\end{proof}%need figure??

\begin{lemma}\label{lemma-others4}
Let $\alpha \neq B$.
If Algorithm~\ref{algo:Ren4} starts with $\{ \ell(r,t_c), \ell(s,t_c)\}=\{\alpha,\beta\}$ such that $\alpha \rightarrow \beta$, 
for any schedule of two robots after $t_c$,
there is a cs-time $t^* (\geq t_c)$ such that 
$\{\ell(r,t^*), \ell(s,t_c)\} =\{ B, C\}$ and 
$dis(p(r,t^*),p(s,t^*)) \leq dis(p(r,t_c),p(s,t_c))$.
\end{lemma}
\begin{proof}
We can assume that $\ell(r,t_c)=\alpha$ and $\ell(s,t_c)=\beta$.
There are three cases, $\alpha=A,C$ and $D$.

(I) ($\alpha=A$) Let $t_r=t^+_c(r,\mathit{LC})$. 
If $s$ does not become active until $t_r+1$,
$t_r+1$ is a cs-time $t^*$ satisfying the conditions of this lemma.
Otherwise,
If $s$ performs a cycle at $t (\leq t_r)$, 
it changes its color to $B$ at $t+1$ and will move to the midpoint between $p(r,t)$ and $p(s,t)$. 
Let $t_s = t^-_r(s,\mathit{LC}) (\leq t_r)$. Note that $\ell(s,t_s)=B$.
If $t^+_s(s,\mathit{M_E}) \leq t_r$, then $t_r+1$ becomes a cs-time $t^*$
such that $\ell(r,t^*)=C$, $\ell(s,t^*)=B$ and $dis(p(r,t^*),p(s,t^*)) < dis(p(r,t_c),p(s,t_c))$.

If $t^+_s(s,\mathit{M_E}) > t_r$, $t^+_s(s,\mathit{M_E})+1$ becomes a cs-time $t^*$
satisfying the conditions of this lemma, since $r$ does not change its color and stays at position $p(r,t_s)$, 
even if it performs cycles between $t_s$ and  $t^+_s(s,\mathit{M_E})$.

(II) ($\alpha=C$) Similar to the case (I), we can show that 
there is a cs-time $t^* (\geq t_c)$ such that 
$\{\ell(r,t^*), \ell(s,t_c)\} =\{ A, B\}$ and 
$dis(p(r,t^*),p(s,t^*)) \leq dis(p(r,t_c),p(s,t_c))$. Then the lemma holds by using the case (I).

(III) ($\alpha=D$) Let $t_r = t^+_c(r,LC)$. By Lemma~\ref{generalprop4},
$\ell(r,t_r)=D$ and $\ell(s,t_r)=A$. Robot $r$ changes its color to $B$ at $t_r+1$ and will move to the midpoint of $p(r,t_r)$ and $p(s,t_r)$. Letting $t'=t^+_r(r,\mathit{M_E})$, $s$ does not change its color and stays at the position $p(s,t_r)$. Thus, $t'$ becomes a cs-time such that 
$\{\ell(r,t'), \ell(s,t')\} =\{ A, B\}$ and 
$dis(p(r,t'),p(s,t')) \leq dis(p(r,t_c),p(s,t_c))$. Then the lemma holds by using the case (I).
\end{proof}

\begin{lemma}\label{lemma-sameLC4}
Let $\ell(r,t_c)=\ell(s,t_c)$. 
If all $\mathit{LC}$-operations of $r$ and $s$ are performed at the same times,
there is a cs-time $t^* (\geq t_c)$ such that $dis(p(r,t^*),p(s,t^*))=0$.
\end{lemma}
\begin{proof}
Let $t_i(i=1,2,3,\ldots)$ be the times $r$ and $s$ perform $\mathit{LC}$-operations simultaneously, and let $\ell(r,t_i)=\ell(s,t_i)=\alpha$ and $\alpha \rightarrow \beta$. Note that $\mathit{Move}$-operations of  both robots are performed
between $t_i$ and $t_{i+1}$.
If $\ell(r,t_i)=\ell(s,t_i)=\alpha$, $r$ and $s$ change their color to $\beta$ and $\ell(r,t)=\ell(s,t)=\beta (t_i <t \leq t_{i+1}])$.
If $\alpha=B,D$, the two robots stay at the positions of $t_i$.
If $\alpha=C$, the two robots swap positions compared to $t_i$.
If $\alpha=A$, the two robots move to the midpoint of their positions.

Therefore, when $\ell(r,t_i)=\ell(s,t_i)=A$, $dis(p(r,t_{i+1}), p(s,t_{i+1}))=0$ if
$dis(p(r,t_{i}), p(s,t_{i})) \leq 2\delta$, the lemma holds.
Otherwise, $dis(p(r,t_{i+1}), p(s,t_{i+1})) \leq dis(p(r,t_{i}), p(s,t_{i}))-2\delta$.
This reduction occurs whenever $\ell(r,t_j)=\ell(s,t_j)=A$ and the distance between $r$ and $s$ will becomes $0$.
\end{proof}

\begin{lemma}\label{lemma-difLC4}
Let $\ell(r,t_c)=\ell(s,t_c)$. If there is a different time at which $\mathit{LC}$-operations of $r$ and $s$ are performed, 
for any schedule of two robots after $t_c$, 
there is a cs-time $t^* (\geq t_c)$ and there are colors $\alpha$ and $\beta$ such that
$\{\ell(r,t^*), \ell(s,t^*)\}=\{\alpha, \beta\}$ and $\alpha \rightarrow \beta$.
%$\{\ell(r,t^*), \ell(s,t^*)\}=\{B, C\}$.
\end{lemma}%must be cahnged!!
%\alpha \rightarrow \beta \rightarrow \gamma: \alpha=A,C  ---> (\beta,\gamma)
%alpha=B,D ---. (\beta,\alpha)
\begin{proof}
Let $t_r$ and $t_s$ be times of $\mathit{LC}$-operations  performed by $r$ and $s$ such that $t_r \neq t_s$ and these are the first different times of $\mathit{LC}$-operations performed by $r$ and $s$.  Wlog, assume that $t_r < t_s$ and  let $t^-_r(r, \mathit{LC})=t^-_s(s,\mathit{LC})=t'$.
There are four cases according to the colors of the robots at time $t'$.

(I) $\ell(r,t')=\ell(s,t')=A$. Robots $r$ and $s$ change their color to $B$ at $t'+1$ and they will move to the midpoint of $p(r,t')$ and $p(s,t')$. 
Since $\ell(s,t_r)=B$, $r$ changes its color to $C$ at $t_r+1$ and stays at position $p(r,t_r)$. Since $r$ does not change its color between $t_r+1$ and $t_s$, %a cs-time $t^* (\geq t_c)$ and there are colors $\alpha$ and $\beta$ such that
%$\{\ell(r,t^*), \ell(s,t^*)\}=\{\alpha, \beta\}$ and $\alpha \rightarrow \beta$.
then $t_s$ is a cs-time $t^* (\geq t_c)$ such that
$\{\ell(r,t^*), \ell(s,t^*)\}=\{B, C\}$ ($B \rightarrow C$).
%$s$ observes the color  $C$ of $r$ at the time $t_s$ and changes its color to $D$ at $t_s+1$ and will move to the position of $p(r,t_s)$.

(II) $\ell(r,t')=\ell(s,t')=B$ Robots $r$ and $s$ change their color to $C$ at $t'+1$ and they will stay until $t_r$.
Since $\ell(s,t_r)=C$, $r$ changes its color to $D$ at $t_r+1$ and will move to position $p(s,t_r)$.
If $t_s \leq t^+_r(r,\mathit{M_E})$, $t^+_r(r,\mathit{M_E})+1$ is a a cs-time $t^* (\geq t_c)$ such that
$\{\ell(r,t^*), \ell(s,t^*)\}=\{D, A\}$. Otherwise, $t^+_r(r,\mathit{M_E})+1$ is a a cs-time $t^* (\geq t_c)$ such that
$\{\ell(r,t^*), \ell(s,t^*)\}=\{C, D\}$.

(III) $\ell(r,t')=\ell(s,t')=C$ and 
(IV) $\ell(r,t')=\ell(s,t')=D$ can be proved similar to cases (I) and (II), respectively.
\end{proof}

\begin{lemma}\label{lemma-notACBD4}
Let $\ell(r,t_c)=\ell(s,t_c)$. If Algorithm~\ref{algo:Ren4} is performed starting from $t_c$, there does not exist any cs-time $t^* (\geq t_c)$  such that $\{\ell(r,t^*), \ell(s,t^*)\}=\{A,C\}$ or $\{ \ell(r,t^*), \ell(s,t^*)\}=\{B, D\}$.
\end{lemma}
\begin{proof}
It can be verified in the proofs of the above lemmas.
\end{proof}

\begin{theorem}
Rendezvous is solved by QSS-Rendezvous-with-4-colors($LC$-atomic ASYNC, Non-Rigid, any) with $\ell(r,t_0)=\ell(s,t_0)$.
It is a quasi-self-stabilizing \ML-algorithm.
\end{theorem}
\begin{proof}
The theorem is derived from Lemmas~\ref{lemma-zero4}-\ref{lemma-notACBD4}. 
\end{proof}

Algorithm~\ref{algo:Ren5} also satisfies similar properties of Lemmas~\ref{lemma-zero4}-\ref{lemma-notACBD4} 
(Lemmas~\ref{lemma-zero5}-\ref{lemma-notACBD5}) and
it can be also shown to be a self-stabilizing \ML-algorithm.
In Algorithm~\ref{algo:Ren4}, two color pairs $\{A,C\}$ and $\{B, D\}$ of $r$ and $s$ cannot be reached from any initial configuration with  same colors  (Lemma~\ref{lemma-notACBD4}). However, it cannot achieve Rendezvous from the initial configuration $\{A,C\}$ or $\{B, D\}$ (Lemma~\ref{ImpAlg4}), since repetitions of $\{A,C\}$ and $\{B, D\}$ never attain Rendezvous. This is the reason why Algorithm~\ref{algo:Ren4} is not self-stabilizing.
On the other hand, we can show that Algorithm~\ref{algo:Ren5} is self-stabilizing. In fact,
it can solve Rendezvous from the initial configurations  $\{A,C\}$, $\{B, D\}$, $\{C, E\}$, $\{D, A \}$ or $\{E, B\}$ as expressed in a  following lemma (Lemma~\ref{lemma-notACBD5}). 
If these configurations repeat, since the repetition contains $\{C, E\}$, Rendezvous succeeds. Otherwise, any configuration can reach  some configuration  $\{\alpha, \beta\}$ ($\alpha \rightarrow \beta$).
It can be proved similarly to Lemmas~\ref{lemma-sameLC4} and \ref{lemma-difLC4}. 

%Lemmas12-18 can be omitted!!!
\begin{lemma}\label{lemma-zero5}
If $dis(p(r,t_c),p(s,t_c))=0$ and Algorithm~\ref{algo:Ren5} is executed starting from $t_c$, $dis(p(r,t),p(s,t))=0$ for any $t \geq t_c$.
\end{lemma}

\begin{lemma}\label{generalprop5}
Let $\alpha=B$ and $\beta=C$ or $\alpha=D$ and $\beta=A$.
If $\ell(r,t_c)=\alpha$ and $\ell(s,t_c)=\beta$ in Algorithm~\ref{algo:Ren5}, then  % such that $\alpha \rightarrow \beta$ and let $t_1=t^+_c(r,LC)$, then
$\ell(s,t)=\beta$  and $p(s,t)=p(s,t_c)$ for any $t(t_c \leq t \leq t_1=t^+_c(r,\mathit{LC}))$.
\end{lemma}

\begin{lemma}\label{lemma-BC5}
If Algorithm~\ref{algo:Ren5} starts with $\{\ell(r,t_c), \ell(s,t_c)\} =\{B, C\}$,  %such that $\alpha \rightarrow \beta$,
for any schedule of two robots after $t_c$,
there is a cs-time $t^* (\geq t_c)$ such that $dis(p(r,t^*),p(s,t^*))=0$,
or $dis(p(r,t^*),p(s,t^*)) \leq dis(p(r,t_c),p(s,t_c))-\delta$ 
and $\{\ell(r,t^*), \ell(s,t^*)\} =\{C, D\}$ or $\{\ell(r,t^*), \ell(s,t^*)\} =\{D, A\}$ .
\end{lemma}

\begin{lemma}\label{lemma-others5}%must check!!
Let $\alpha \neq B$.
If Algorithm~\ref{algo:Ren5} starts with $\{ \ell(r,t_c), \ell(s,t_c)\}=\{\alpha,\beta\}$ such that $\alpha \rightarrow \beta$, 
for any schedule of two robots after $t_c$,
there is a cs-time $t^* (\geq t_c)$ such that 
$\{\ell(r,t^*), \ell(s,t_c)\} =\{ B, C\}$ and 
$dis(p(r,t^*),p(s,t^*)) \leq dis(p(r,t_c),p(s,t_c))$.
\end{lemma}

\begin{lemma}\label{lemma-sameLC5}
Let $\ell(r,t_c)=\ell(s,t_c)$ in Algorithm~\ref{algo:Ren5}. 
If all $\mathit{LC}$-operations of $r$ and $s$ are performed at the same times,
there is a cs-time $t^* (\geq t_c)$ such that $dis(p(r,t^*),p(s,t^*))=0$.
\end{lemma}

\begin{lemma}\label{lemma-difLC5}
Let $\ell(r,t_c)=\ell(s,t_c)$ in Algorithm~\ref{algo:Ren5}. If there is a different time at which $\mathit{LC}$-operations of $r$ and $s$ are performed, 
for any schedule of two robots after $t_c$, 
there is a cs-time $t^* (\geq t_c)$ and there are colors $\alpha$ and $\beta$ such that
$\{\ell(r,t^*), \ell(s,t^*)\}=\{\alpha, \beta\}$ and $\alpha \rightarrow \beta$.
%$\{\ell(r,t^*), \ell(s,t^*)\}=\{B, C\}$.
\end{lemma}%must be cahnged!!
%\alpha \rightarrow \beta \rightarrow \gamma: \alpha=A,C  ---> (\beta,\gamma)
%alpha=B,D ---. (\beta,\alpha)

\begin{lemma}\label{lemma-notACBD5}
Let $\ell(r,t_c)=\ell(s,t_c)$. If Algorithm~\ref{algo:Ren5} is performed starting from $t_c$, there exist no cs-time $t^* (\geq t_c)$  such that $\{\ell(r,t^*), \ell(s,t^*)\}=\{\alpha, \gamma\}$, where $\alpha \rightarrow \beta$ and $\beta \rightarrow \gamma$.
\end{lemma}

%It can be proved in the similar way that Lemmas~\ref{lemma-zero4}-\ref{lemma-difLC4} hold for Algorithm~\ref{algo:Ren5}.
%Also the next lemma holds for Algorithm~\ref{algo:Ren5}.

\begin{lemma}\label{lemma-ACetc}
Let $\{\ell(r,t_0), \ell(s,t_0)\}=\{\alpha, \gamma \}$, where $\alpha \rightarrow \beta$ and $\beta \rightarrow \gamma$.
Algorithm~\ref{algo:Ren5} can solve Rendezvous from any initial configuration $\{\ell(r,t_0), \ell(s,t_0)\}=\{\alpha, \gamma \}$.
\end{lemma}

To Lemmas~\ref{lemma-zero5}-\ref{lemma-ACetc} follow the next theorem.

\begin{theorem}
Rendezvous is solved by SS-Rendezvous-with-5-colors($\mathit{LC}$-atomic ASYNC, Non-Rigid, any).
It is a self-stabilizing \ML-algorithm.
\end{theorem}

%\begin{wrapfigure}{R}{0.5\textwidth} 
%\vspace{-20pt}
% \begin{center}
%    %\includegraphics[scale=0.5,width=4cm]{Fig1.eps}
%    \includegraphics[scale=0.5]{Fig1.eps}
%    \caption{Transition Graph for Full-Light-Gather.}
%    \label{fig:Full-light}
%  \end{center}
%  \vspace{-20pt}
%  \vspace{1pt}
%\end{wrapfigure} 
%\begin{figure*}%{R}{0.5\textwidth} 
%\vspace{3cm}
% \begin{center}
%    \includegraphics[scale=0.8]{theorem3.eps}
    %\includegraphics[scale=0.8]{counter-ex.eps}
    %\includegraphics[scale=0.5]{Fig1.eps}
%    \caption{Several cases in the proof of Theorem~\ref{AsyncRigidA}}
%    \label{figCaseTheorem3}
%  \end{center}
%  \vspace{-20pt}
%  \vspace{1pt}
%\end{figure*} 

%\begin{figure*}%{R}{0.5\textwidth} 
%\vspace{3cm}
% \begin{center}
%    \includegraphics{counter-ex.eps}
    %\includegraphics[scale=0.5,width=4cm]{counter-ex.eps}
    %\includegraphics[scale=0.5]{Fig1.eps}
%    \caption{Rendezvous(ASYNC, Rigid, B) cannot solve Rendezvous in general.}
%    \label{figExecutionAsyncRigidB}
%  \end{center}
%  \vspace{-20pt}
%  \vspace{1pt}
%\end{figure*} 

%\begin{theorem}\label{ASYNCNonRigidDeltaA}
%RedezvousWithDelta(ASYNC, Non-Rigid($+\delta$), $A$) solves Rendezvous.
%\end{theorem}

\section{Concluding Remarks}
\label{sec:conclusion}

We have shown that Rendezvous can be solved by \ML-algorithms in $\mathit{LC}$-atomic ASYNC with the optimal number of colors of external-lights
in the following cases.
(1) Rigid and non-quasi-self-stabilizing, (2) Non-Rigid and quasi-self-stabilizing, and (3) Non-Rigid and self-stabilizing.
%if Rigid  or Non-Rigid($+\delta$) movement is assumed.
%We have also shown that Rendezvous can be solved in ASYNC and Non-Rigid with  the optimal number of colors of lights if ASYNC is LC-atomic.
%Although we conjecture that Rendezvous cannot be solved in ASYNC and Non-Rigid with $2$ colors,
%it is still open whether it can be solved or not.
%We have shown Gathering algorithms by mobile robots with lights in SSYNC and CENT schedulers and
%we have obtained some relationship between the power of lights (full, external and internal) and assumptions of robots' moving (rigid and non-rigid).
%Interesting open questions are determining the precise computational relationship between external-light and internal-light and constructing Gathering %algorithms by external-light and internal-light in ASYNC.
%gathering algorithm, model extension,
%Comparison between internal and external

%\vspace{-3mm}
\paragraph*{Acknowledgment}
This work is supported in part by KAKENHI no. 17K00019.% and 15K00011.%17K00019(CPS-robot)

%}%smaill size

\newpage
\appendix

\end{document}